\documentclass[11pt]{article}
\usepackage[a4paper,margin=1in]{geometry}
\usepackage{amsmath,amssymb,amsthm,mathtools}
\usepackage{microtype}
\usepackage[colorlinks=true,linkcolor=blue,citecolor=blue,urlcolor=blue]{hyperref}
\usepackage{xcolor}
\usepackage{graphicx}
\usepackage{adjustbox}
\usepackage[shortlabels]{enumitem}
\usepackage{tikz}
\usepackage{subcaption}
\usepackage[percent]{overpic}
\usepackage[normalem]{ulem}
\usepackage{comment}
\usepackage{subcaption}

\usepackage[
    backend=biber,
    citestyle=numeric-comp,
    sorting=none,
    doi=true,
    eprint=true,
    url=false,
    isbn=false,
    maxbibnames=7,
    uniquename=init,
    giveninits=true
]{biblatex}
\usepackage{hyperref}

\AtEveryBibitem{
    \clearfield{month}
    \clearfield{date}
    \clearfield{day} 
    \clearfield{primaryClass}
}

\usetikzlibrary{calc,arrows.meta,angles,quotes,patterns}

\newcommand{\ConeShellFigure}[3][\linewidth]{%
  \begin{overpic}[width=#1]{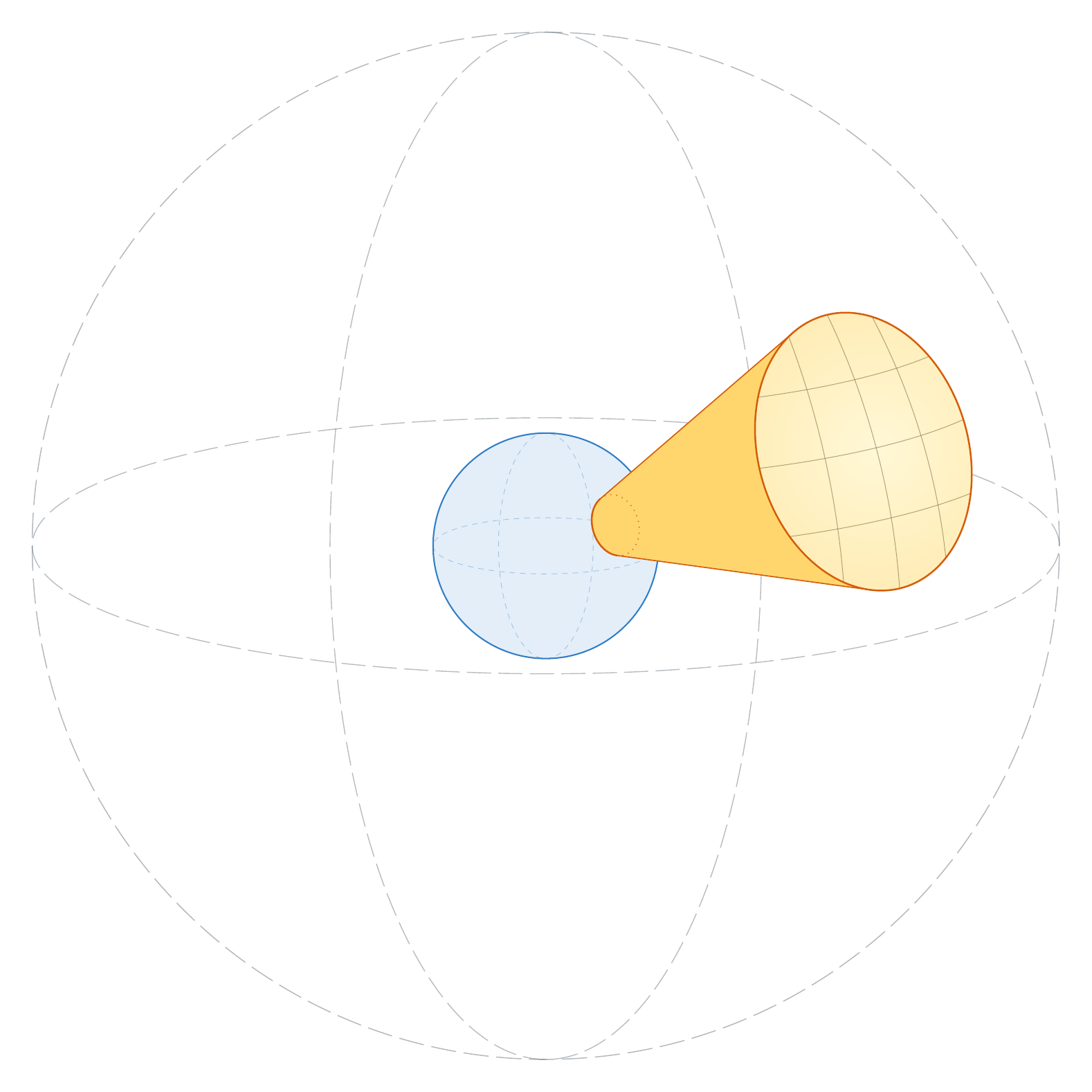}%
    \put(48.0,50){\makebox(0,0){#2}}%
    \put(65.6,55.6){\makebox(0,0){#3}}%
  \end{overpic}%
}

\definecolor{darkblue}{rgb}{0.,0.,0.4}
\definecolor{darkred}{rgb}{0.5,0.,0.}
\definecolor{BlueViolet}{RGB}{138,43,226}
\definecolor{SkyBlue}{RGB}{30,144,255}
\definecolor{DarkGreen}{RGB}{0,100,0}
\definecolor{exteriorcolor}{RGB}{210,210,210}
\definecolor{conepale}{RGB}{245,225,150}
\definecolor{sphereblue}{RGB}{70,130,190}
\definecolor{wiregray}{RGB}{120,120,120}
\definecolor{coneedge}{RGB}{180,140,40}

\newtheorem{theorem}{Theorem}[section]

\newtheorem{lemma}[theorem]{Lemma}

\newtheorem{introtheorem}{Theorem} %letter-labeled thm env for intro

\theoremstyle{remark}

\newtheorem*{remark*}{Remark}

\newcommand{\A}{\mathcal A}

\newcommand{\F}{\mathcal F}

\renewcommand{\H}{\mathcal H}

\newcommand{\supp}{\operatorname{supp}}

\newcommand{\id}{\mathrm{id}}

\newcommand{\ox}{\otimes}
\newcommand{\eps}{\varepsilon}
\newcommand{\ZZ}{\mathbb Z}

\newcommand{\RR}{\mathbb R}
\newcommand{\CC}{\mathbb C}
\renewcommand{\Hat}[1]{\widehat{#1}}
\newcommand{\ip}[2]{\langle #1, #2 \rangle}
\newcommand{\placeholder}{\,\cdot\,}

\DeclareMathOperator{\tr}{Tr}

\newcommand\oo\infty

\title{An entropic characterization of Haag duality }
\author{Ruizhi Liu$^{1,2}$ \and Lauritz van Luijk$^{1,3}$} 
\date{{\small
$^1$Perimeter Institute for Theoretical Physics\\
$^2$Department of Mathematics and Statistics, Dalhousie University, Nova Scotia, Canada\\
$^3$Institute for Quantum Computing, University of Waterloo, Ontario, Canada}  \\[11pt]
\today}

\begin{document}

\maketitle

\begin{abstract}
Motivated by Haag duality in quantum spin systems, we show that Haag duality for a pair of commuting factors generated by increasing sequences of finite-dimensional matrix algebras is equivalent to the asymptotic vanishing of a suitable conditional mutual information.

For spin systems, let an annulus separate a finite region from the exterior. Haag duality then holds for a region $A$ if and only if the mutual information of the inner region and the exterior, conditioned on the part of $A$ inside the annulus, vanishes as the outer radius becomes large. 
As a corollary, assuming a strict area law with subleading corrections, Haag duality holds for single cones, and it holds for unions of two or more disjoint cones precisely when the topological entanglement entropy vanishes. 
For translation-invariant pure states on spin chains, half-chain Haag duality is equivalent to vanishing entropy density.
\end{abstract}

\section{Introduction}

Haag duality is an important algebraic property characterizing bipartitions of quantum systems with infinitely many degrees of freedom in a fixed sector \cite{haag2012local}.
It asserts that the von Neumann algebras $A,B$ modeling the two subsystems in the sector $\H$ under consideration are commutants of each other
\begin{equation}
    A = B', \qquad X' = \{ y\in B(\H) \,:\, xy=yx\ \forall x\in X\}.
\end{equation}
Haag duality appeared first in the study of algebraic quantum field theory \cite{haag2012local} and has been studied extensively since.
To this day, it is much studied in high-energy physics \cite{jensen_generalized_2023,shao_additivity_2025,harlow_disjoint_2025,faulkner_asymptotically_2022,holfester_algebraic_2026,gui_minkowskian_2026,galanda_haag_2026,naaijkens_local_2026}, quantum matter \cite{keyl_entanglement_2006,naaijkens_haag_2012,fiedler_haag_2015,Naaijkens_2017,jones_local_2025,vanLuijk2025,vanluijk2025largescalestructureentanglementquantum,ogata_haag_2025,wallick_nonabelian_2026,Kato-upcoming}, and quantum information theory of systems with infinitely many degrees of freedom \cite{summers_vacuum_1985,summers_maximal_1987,van_luijk_schmidt_2024,vanLuijk2024Embezzlement,van_luijk_entanglement_2025,vanLuijk2025Pure,van_luijk_uniqueness_2026,van_luijk_quantum_2026}.

Perhaps the most important application of Haag duality and its variants is that it gives rise to superselection sector theory, both in quantum field theory \cite{fredenhagen_superselection_1989,DoplicherHaagRoberts1971,DoplicherHaagRoberts1974} and in quantum matter \cite{naaijkens_haag_2012,naaijkens_localized_2011,Ogata_2022_tensorcat,ogata_classification-chains_2022,bhardwaj_superselection_2025,Kato-upcoming}.
In the former, the categories constructed from sector theory classify particle statistics, localized charges and their fusions; in the latter these categories classify topological excitations, together with their fusion and braiding relations.

Haag duality was recently shown to have a quantum-information theoretic interpretation \cite{van_luijk_uniqueness_2026}: It is equivalent to the fundamental principle that two global pure states can be connected by unitaries from a subsystem if and only if they have the same marginal states on the complementary subsystem.
In follow-up work \cite{van_luijk_quantum_2026}, another information-theoretic characterization was given: Both parties can steer each ensemble of the marginal state of the respective other party.

Another broader role that Haag duality plays in the literature is that of a simplifying mathematical assumption because it allows one to use the powerful toolkit of mathematical results connecting a von Neumann algebra and its commutant in a direct way.
This toolkit includes techniques such as Radon-Nikodym theorems \cite{takesaki1,belavkin_radon-nikodym_1986,vaes_radon-nikodym_2001}, standard representations \cite{haagerup_standard_1975}, Uhlmann's theorem \cite{uhlmann_transition_1976,alberti_note_1983}, Tomita-Takesaki theory \cite{takesaki2,Summers2005tomita}, spatial derivatives \cite{connes_spatial_1980}, and many more.
For instance, assuming Haag duality, it was shown that the classification of von Neumann algebras into types and subtypes is equivalent to a list of operational entanglement properties \cite{vanLuijk2025Pure,vanLuijk2024Embezzlement,vanLuijk2024PRL,van_luijk_entanglement_2025}.

A common theme in the literature is that Haag duality is often \emph{assumed}, but only rarely proved.
A notable exception are conformal field theories on the circle, where positive-energy and Möbius covariance imply Haag duality for intervals \cite{gabbiani_operator_1993}.
In most cases where Haag duality has successfully been shown, the models are exactly solvable and the proofs are hard to generalize beyond the studied class of models \cite{ogata_haag_2025,matsui2011boundedness,naaijkens_haag_2012,fiedler_haag_2015,vanLuijk2025,naaijkens_local_2026}.
It is believed that half-chain Haag duality holds for general translation-invariant pure states on quantum spin chains.%
\footnote{\label{footnote:wrong-hd-proofs}
Two alleged proofs have been published \cite{KEYL_2008,mohari_translation_2014}.
The proof in \cite{KEYL_2008} has a mistake, as reported by two of the authors in \cite{Matsui2010,vanLuijk2025Pure}.
The paper \cite{mohari_translation_2014} sets out to fix this mistake, but several steps in its argument remain unclear.
We, therefore, regard the conjecture as open.
}
The recent paper \cite{wallick_nonabelian_2026} constructs the first example of a unique ground state of a local Hamiltonian in a 2D spin system for which Haag duality is false for cone-like regions.

\begin{figure}[t!]
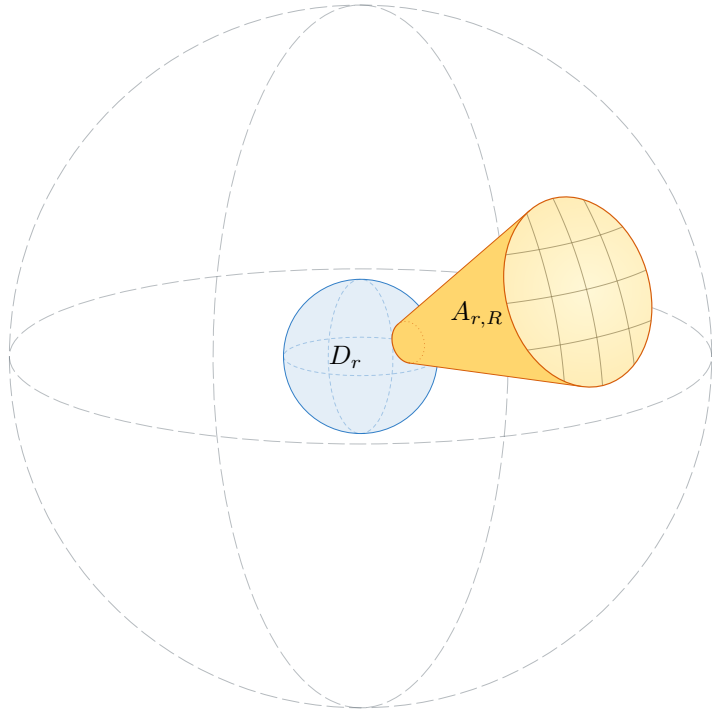

\centering
\clipbox{0.2cm 0.2cm 0.2cm 0.2cm}{% left bottom right top
    \ConeShellFigure[0.62\linewidth]
      {\small $D_r$}
      {\small $A_{r,R}$}%
      % {\small $D_R^c$}
  }
  \caption{
    Geometry in three dimensions for a cone $A$:
    We are interested in the mutual information between a ball $D_r$ of radius $r$ (blue) and the exterior of a ball of radius $R> r$, conditioned on $A_{r,R}$ (yellow), the part of the cone in the spherical shell of inner radius $r$ and outer radius $R$.
    Haag duality holds if and only if the mutual information is negligible in the regime $R\gg r$.
    }
  \label{fig:3D-case}
\end{figure}

In this paper, we show that Haag duality in quantum spin systems can be seen as an asymptotic quantum Markovianity property.
To explain the result in its general form, we consider an abstract algebraic setting.
We consider a pair of commuting von Neumann factors $A,B$ on a Hilbert space $\H$, 
and we assume that $A$ and $B$ are generated by increasing families of finite matrix algebras $(A_n)_n$ and $(B_n)_n$ and that $A\vee B = B(\H)$. 
In a fixed sector of a quantum spin system, $A$ and $B$ are the von Neumann algebras associated with a region and its complement, and $A_n$, $B_n$ are approximations on finite regions.
For $n\ge m$, we denote the relative commutants of $A_m\subset A_n$ and $B_m\subset B_n$ by
\begin{equation}
    A_{m,n} := A_m'\cap A_n ,\qquad B_{m,n} := B_m'\cap B_n.
\end{equation}
We introduce the notation
\begin{equation}
    D_n := A_n\vee B_n.
\end{equation}
It follows that $A_{m,n}$, $B_{m,n}$, $D_n'$ and $D_m$ are pairwise commuting type I factors, which generate $B(\H)$, see Sec.~\ref{sec:proof-abstract} for a proof.
Relative to a global pure state $\Omega\in\H$, we can define the conditional mutual information as%
\footnote{Relative to a mixed state $\omega$, the conditional mutual information is normally defined as $S(D_m \vee A_{m,n})_\omega + S(D_n'\vee A_{m,n})_\omega - S(A_{m,n})_\omega - S(D_m \vee A_{m,n} \vee D_n')_\omega$. 
The two definitions agree relative to a pure state since $B(\H)= A_{m,n}\vee B_{m,n}\vee D_m \vee D_n'$ implies $S(B_{m,n})_\Omega=S(D_m \vee A_{m,n} \vee D_n')_\Omega$ and $S(D_n'\vee A_{m,n})_\Omega=S(D_m \vee B_{m,n})_\Omega$.
We use eq.~\eqref{eq:CMI-entropy-rewrite} because it only involves finite-dimensional algebras.}
\begin{align}\label{eq:CMI-entropy-rewrite}
    I(D_m : D_n' \mid A_{m,n})_\Omega 
    := S(D_m \vee A_{m,n})_\Omega + S(D_m\vee B_{m,n})_\Omega - S(A_{m,n})_\Omega - S(B_{m,n})_\Omega,
\end{align}
where $S(X)_\Omega$ is the von Neumann entropy of the reduced state on $X$.
We can now state our abstract entropic theorem, roughly saying that Haag duality is equivalent to $D_m\mid A_{m,n} \mid D_n'$ being an approximate quantum Markov chain in the regime $n\gg m$:

\begin{introtheorem}\label{introthm:abstract}
    The following are equivalent:
    \begin{enumerate}[(a)]
        \item Haag duality $A=B'$;
        \item For every unit vector $\Omega\in \H$ and every $m$, 
            \begin{equation}\label{eq:CMI-outer-limit-only}
                \lim_{n\to\oo} \ I( D_m : D_n' \mid A_{m,n})_\Omega = 0.
            \end{equation}
        \item For some unit vector $\Omega\in\H$, the left-hand side of \eqref{eq:CMI-outer-limit-only} goes to zero as $m\to\oo$.
    \end{enumerate}
\end{introtheorem}

We now sketch the case of quantum many-body systems.
Let $A$ be a region in a lattice in arbitrary spatial dimension and let $B$ be its complement.
For the finite approximations of $A$ and $B$, we consider their intersections $A_r:=A\cap D_r$ and $B_r=B\cap D_r$ with a concentric family $\{D_r\}_{r>0}$ of balls with $r$ being the radius. 
% For a larger radius $R>r$, let $A_{r,R}$ be the part of $A$ in the spherical shell of inner radius $r$ and outer radius $R$.
Our theorem now says that Haag duality holds for $A$ in the sector of a pure state $\Omega$ if and only if conditioning on the part of $A$ in a shell of inner radius $r$ and outer radius $R$ removes almost all correlations between the inner radius-$r$ ball and the exterior of the outer radius-$R$ ball, in the regime $R\gg r$, see Fig.~\ref{fig:3D-case}.
In Section~\ref{sec:many-body}, we discuss applications of Theorem~\ref{introthm:abstract} to quantum many-body systems in detail.

Our result shows that deciding Haag duality is as hard as computing the limit of a certain combination of entanglement entropies of finite regions.
Indeed, the conditional mutual information that quantifies approximate Markovianity is simply a sum of four finite entanglement entropies.
Thereby, our result gives a model-independent method for checking Haag duality.
As a consequence, we find that Haag duality holds for all cones in models satisfying a strict area law with subleading corrections, and show that the validity of Haag duality for disjoint unions of cones is connected to the topological entanglement entropy.
In one spatial dimension, we find that half-chain Haag duality holds for translation-invariant models whenever the entropy density is zero.

The argument establishing Theorem~\ref{introthm:abstract} is based on a duality formula for the conditional entropy, connecting an algebra and its commutant, and approximate Petz recovery.
We express the double limit of the conditional mutual information as a sum of conditional entropies.
If Haag duality holds, the duality formula asserts that these cancel.
For the converse, we rewrite this sum as the defect in the data-processing inequality of a suitable mutual information.
The approximate vanishing of this defect implies approximate Petz recovery that allow us to approximate operators in $B'$ with operators in $A$, establishing $B'\subset A$ and, hence, $B'=A$.

The paper is organized as follows.
Section \ref{sec:many-body} discusses physical applications in many-body systems.
In two-dimensional systems, we discuss Haag duality in the context of an area law with subleading corrections, the approximate entanglement bootstrap program, and non-abelian anyons.
In one-dimensional systems, we show that half-chain Haag duality is equivalent to a vanishing entropy density, and how this gives rise to an improved Lieb-Schultz-Mattis theorem.
In Section~\ref{sec:condEnt-duality}, we derive a duality formula for the conditional entropy from the literature.
The latter plays a crucial role in Section \ref{sec:proof-abstract}, where the proof of the abstract Haag duality theorem is given.
Appendix~\ref{sec:appendix-recovery} derives an approximate Petz recovery theorem without faithfulness assumptions from the literature.

\section{Applications to quantum many-body systems}\label{sec:many-body}

We consider a pure state $\omega$ on a quantum system of spins with local Hilbert space dimension $d\ge2$ supported on a lattice $\Lambda\subset \RR^k$ in $k\ge1$ spatial dimensions.
Mathematically, $\omega$ is a pure state on the quasi-local C*-algebra
\begin{equation}
    \A_\Lambda := \bigotimes_\Lambda M_d(\CC).
\end{equation}
Let $(\pi,\H,\Omega)$ be the GNS representation of $\omega$.
We denote the von Neumann algebra associated with a region $X\subset \Lambda$ by
\begin{equation}
    M_X := \pi(\A_X)''\subset B(\H),
\end{equation}
where $\A_X=\otimes_X M_d(\CC)\subset \A_\Lambda$.
Irreducibility of the GNS representation implies that, for an arbitrary region $X$, $M_X$ is a factor and that $M_X \vee M_{X^c}=B(\H)$, where $X^c:=\Lambda\setminus X$ denotes the complement.

Let us fix a specific region $A \subset \Lambda$, e.g., a cone stretching out to infinity, and let $B:=A^c$ be its complement.
We are interested in whether Haag duality holds for $A$, that is, whether $M_A= M_B'$.
We now consider an increasing family of finite regions that exhaust $\Lambda$.
For concreteness, we take balls of increasing radii 
\begin{equation}
    D_r = \big\{ x \in \Lambda \ : \ |x-x_0| \le r\big\}, \qquad r>0,
\end{equation}
centered around some point $x_0\in \RR^k$.
The algebra $M_{D_r}$ is a finite type I factor. 
In particular, this implies that Haag duality holds for $D_r$:
\begin{equation}\label{eq:HD-for-Dr}
    (M_{D_r})' = M_{D_r^c}.
\end{equation}
Setting $A_r := D_r\cap A$ and $B_r:= D_r\cap B$, we obtain increasing families of matrix algebras $(M_{A_r})_r$ and $(M_{B_r})_r$ that generate $M_A$ and $M_B$, allowing us to apply Theorem~\ref{introthm:abstract}.
The relative commutant $M_{A_r}'\cap M_{A_R}$ is then the algebra associated with the region $A_R \setminus A_r$ and similarly for $B$.
The choice of $(D_r)$ as radius-$r$ balls gives
\begin{equation}
    A_{r,R} = \big\{ x\in A \ : \ r<|x-x_0|\le R \big\}
\end{equation}
and analogously for $B_{r,R}$.
The case where $A$ is a cone in three dimensions is visualized in Fig.~\ref{fig:3D-case}, and the case where it is a single cone or a union of disjoint cones in two dimensions is visualized in Fig.~\ref{fig:2D-case}.
Our Haag duality characterization is now concerned with the conditional mutual information $I(D_r :D_R^c \mid A_{r,R})_\omega := I(M_{D_r}:M_{D_R}' \mid M_{A_{r,R}})_\Omega$. 
This shorthand is justified by \eqref{eq:HD-for-Dr}.
In the situation under consideration, Theorem~\ref{introthm:abstract} gives:

\begin{theorem}\label{thm:all-dim}
    Haag duality holds for the region $A$ in the GNS representation of the pure state $\omega$ if and only if 
    \begin{equation}\label{eq:all-dim-inner-limit}
        \lim_{R\to\oo} I(D_r : D_R^c \mid A_{r,R})_\omega
    \end{equation}
    vanishes as $r\to\oo$ if and only if \eqref{eq:all-dim-inner-limit} is zero for each $r>0$.
\end{theorem}

\subsection{Topologically ordered states in 2D}

We analyze our criterion in two-dimensional spin systems.
We consider two cases: $A$ is a single cone, and $A$ is the union of multiple disjoint cones, see Fig.~\ref{fig:2D-case}.
In the single-cone case, Haag duality is, in general, expected to hold, although \cite{wallick_nonabelian_2026} showed that it fails for a non-abelian anyon.
For multiple cones, Haag duality is known to fail in some exactly solvable models by \cite{naaijkens_haag_2012}, and an argument in \cite{van_luijk_uniqueness_2026}, based on the failure of the uniqueness of purifications, shows that it should fail in general topologically ordered systems.

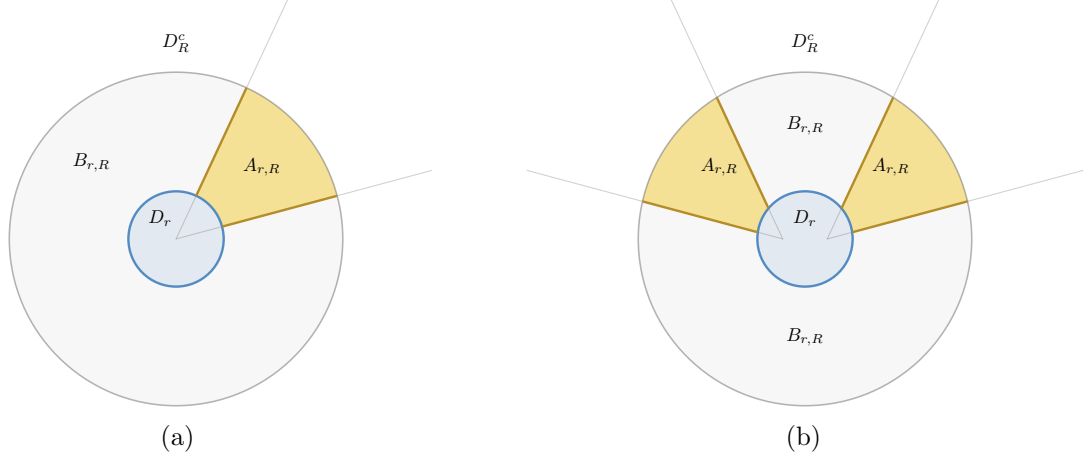
\begin{figure}[ht!]
\centering
% =========================================================
% (a) Single cone with its tip at the center of the disks
% =========================================================
\begin{subfigure}[t]{0.48\textwidth}
\centering
\begin{tikzpicture}[
scale=0.7,
transform shape,
line cap=round,
line join=round,
every node/.style={font=\small}
]
\def\rin{0.9}
\def\rout{3.15}
\def\raylen{5}

% Boundary angles of the cone
\def\rightlower{15}
\def\rightupper{65}

\coordinate (O) at (0,0);
\coordinate (T) at (0,0);

% Outer disk
\fill[exteriorcolor, opacity=0.18]
  (O) circle (\rout);

% Cone restricted to the (r,R)-annulus
\begin{scope}
  \pgfseteorule
  \clip
    (O) circle (\rout)
    (O) circle (\rin);

  \path[fill=conepale, draw=none]
    (T)
    -- ([shift={(\rightlower:\raylen)}]T)
    -- ([shift={(\rightupper:\raylen)}]T)
    -- cycle;
\end{scope}

% Inner disk D_r
\fill[sphereblue, opacity=0.12]
  (O) circle (\rin);

% Light continuations of the cone rays
\draw[wiregray, opacity=0.35, thin]
  (T) -- ([shift={(\rightlower:\raylen)}]T);
\draw[wiregray, opacity=0.35, thin]
  (T) -- ([shift={(\rightupper:\raylen)}]T);
% Circular boundaries
\draw[sphereblue, opacity=0.92, line width=0.9pt]
  (O) circle (\rin);
\draw[wiregray, opacity=0.55, semithick]
  (O) circle (\rout);
% Strong cone boundaries restricted to the annulus
\begin{scope}
  \pgfseteorule
  \clip
    (O) circle (\rout)
    (O) circle (\rin);

  \draw[coneedge, line width=0.9pt]
    (T) -- ([shift={(\rightlower:\raylen)}]T);

  \draw[coneedge, line width=0.9pt]
    (T) -- ([shift={(\rightupper:\raylen)}]T);
\end{scope}

% Labels
\node at (-0.30,0.40) {$D_r$};
\node at (-4.6,0.40) {};
\node at (1.62,1.36) {$A_{r,R}$};
\node at (-1.60,1.45) {$B_{r,R}$};
\node at (0,3.7) {$D_R^c$};
\end{tikzpicture}

\caption{}
\label{fig:2D-one-cone}
\end{subfigure}
\hfill
% =========================================================
% (b) Two cones with separated tips
% =========================================================
\begin{subfigure}[t]{0.48\textwidth}
\centering
\begin{tikzpicture}[
scale=0.7,
transform shape,
line cap=round,
line join=round,
every node/.style={font=\small}
]
\def\rin{0.9}
\def\rout{3.15}
\def\raylen{5}

% Boundary angles of the right-hand cone
\def\rightlower{15}
\def\rightupper{65}

% Boundary angles of the left-hand cone
\def\leftupper{115}
\def\leftlower{165}

\coordinate (O)  at (0,0);

% Separated cone tips inside D_r
\coordinate (TL) at (-0.42,0);
\coordinate (TR) at ( 0.42,0);

% Outer disk
\fill[exteriorcolor, opacity=0.18]
  (O) circle (\rout);

% Both cones restricted to the (r,R)-annulus
\begin{scope}
  \pgfseteorule
  \clip
    (O) circle (\rout)
    (O) circle (\rin);

  % Right-hand cone
  \path[fill=conepale, draw=none]
    (TR)
    -- ([shift={(\rightlower:\raylen)}]TR)
    -- ([shift={(\rightupper:\raylen)}]TR)
    -- cycle;

  % Left-hand cone
  \path[fill=conepale, draw=none]
    (TL)
    -- ([shift={(\leftupper:\raylen)}]TL)
    -- ([shift={(\leftlower:\raylen)}]TL)
    -- cycle;
\end{scope}

% Inner disk D_r
\fill[sphereblue, opacity=0.12]
  (O) circle (\rin);

% Light continuations of the right-hand cone rays
\draw[wiregray, opacity=0.35, thin]
  (TR) -- ([shift={(\rightlower:\raylen)}]TR);

\draw[wiregray, opacity=0.35, thin]
  (TR) -- ([shift={(\rightupper:\raylen)}]TR);

% Light continuations of the left-hand cone rays
\draw[wiregray, opacity=0.35, thin]
  (TL) -- ([shift={(\leftupper:\raylen)}]TL);

\draw[wiregray, opacity=0.35, thin]
  (TL) -- ([shift={(\leftlower:\raylen)}]TL);

% Circular boundaries
\draw[sphereblue, opacity=0.92, line width=0.9pt]
  (O) circle (\rin);

\draw[wiregray, opacity=0.55, semithick]
  (O) circle (\rout);

% Strong cone boundaries restricted to the annulus
\begin{scope}
  \pgfseteorule
  \clip
    (O) circle (\rout)
    (O) circle (\rin);

  % Right-hand cone
  \draw[coneedge, line width=0.9pt]
    (TR) -- ([shift={(\rightlower:\raylen)}]TR);

  \draw[coneedge, line width=0.9pt]
    (TR) -- ([shift={(\rightupper:\raylen)}]TR);

  % Left-hand cone
  \draw[coneedge, line width=0.9pt]
    (TL) -- ([shift={(\leftupper:\raylen)}]TL);

  \draw[coneedge, line width=0.9pt]
    (TL) -- ([shift={(\leftlower:\raylen)}]TL);
\end{scope}

% Labels
\node at (0.0,0.40) {$D_r$};

\node at ( 1.62,1.36) {$A_{r,R}$};
\node at (-1.62,1.36) {$A_{r,R}$};

\node at (0, 2.15) {$B_{r,R}$};
\node at (0,-1.85) {$B_{r,R}$};

\node at (0,3.7) {$D_R^c$};

\end{tikzpicture}

\caption{}
\label{fig:2D-two-cones}
\end{subfigure}

\caption{Geometry in a two-dimensional spin system. On the left, the region $A$ is a single cone, on the right, it is the union of two cones.
Depicted are a disk $D_r$ of radius $r$ (blue), and the portion $A_{r,R}$ of the region $A$ in the annulus of radii $r<R$ (yellow), as well as the remaining portion $B_{r,R}$ of the annulus (grey).
}
\label{fig:2D-case}
\end{figure}

To evaluate our criterion, we assume that the state satisfies an entanglement law of the form \cite{KitaevPreskill2006, LevinWen2006}
\begin{equation}\label{eq:entanglement-law}
  S(X)_\omega =\alpha|\partial X|-\gamma\, b_{0}(\partial X)+o(1),
\end{equation}
where $S(X)_\omega:=S(M_X)_\Omega$ is the von Neumann entropy of the reduced state.
Here, $\alpha>0$ is a constant, $b_{0}(\partial X)$ is the number of connected components of $\partial X$, $\gamma\ge 0$ is a constant called the topological entanglement entropy, and the $o(1)$ term denotes subleading corrections that vanish in the large-$X$ limit.
An entanglement law of this form holds for many topologically ordered ground states \cite{KitaevPreskill2006, LevinWen2006}, but it is not true in general \cite{Zou_2016,Kim2023, levin_physical_2024}. 
Now suppose $A$ is a union of $n\ge 1$ disjoint cones. 
Then
$
    b_0(\partial A_{r,R}) = b_0(\partial B_{r,R}) = n,
$
so that eq.~\eqref{eq:entanglement-law} gives
\begin{align}
    I(D_r : D_R^c \mid A_{r,R})_\omega 
    &= S( D_r \cup A_{r,R})_\omega + S(D_r \cup B_{r,R})_\omega - S(A_{r,R})_\omega - S(B_{r,R})_\omega\nonumber\\
    &= \alpha \cdot 0 + \gamma (-1 -1 + n + n) + o(1) \nonumber\\
    &= 2(n-1)\gamma + o(1).
\end{align}
If, for each of the regions $D_r \cup A_{r,R}$, $D_r \cup B_{r,R}$, $A_{r,R}$, and $B_{r,R}$, the error term $o(1)$ indeed goes to zero as $|\partial X| \to \oo$, then the limit of the conditional mutual information is
\begin{equation}\label{eq:n-cone-CMI}
    \lim_{R\to\oo} \, I(D_r : D_R^c \mid A_{r,R})_\omega =2(n-1)\gamma.
\end{equation}
This limit is zero if and only if $\gamma=0$ or $n=1$.
The argument can be generalized to any spatial dimension as long as eq.~\eqref{eq:entanglement-law} holds.

However, as noted in \cite{Kim2023,levin_physical_2024}, the constant $\gamma$ is \textit{not} invariant under finite-depth quantum circuits.
As discussed in \cite{Kim2023}, a finite-depth quantum circuit can even map a state with $\gamma=0$ to a state with $\gamma>0$.
In combination with our result that Haag duality for unions of two or more disjoint cones is equivalent to $\gamma=0$, this confirms the intuition in
\cite{Ogata_2022_tensorcat} that Haag duality is not preserved under quasi-local unitaries in general.%
\footnote{
Under physically plausible assumptions, the topological entanglement entropy satisfies $\gamma\ge \log D$, where $D$ is the total quantum dimension of the underlying anyon theory \cite{KitaevPreskill2006,LevinWen2006,Kim2023,levin_physical_2024}.
Assuming the so-called distal split property, \cite{naaijkens_haag_2012} shows that $D^2$ is bounded from above by the Jones-Kosaki-Longo subfactor index \cite{Jones1983,Kosaki1986Index,Longo1989} of the inclusion $M_A \subset M_{A^c}'$, where $A=A_1\cup A_2$ is the union of two sufficiently separated cones.
We expect a direct relation between
index and the total quantum dimension to hold. We leave this for future work.
}

\subsubsection{Haag duality and the approximate entanglement bootstrap program}

It is interesting to consider our result in the context of the entanglement bootstrap program \cite{Shi_2019,Shi_2020}, which derives the structure of topological order in 2D spin systems from two microscopical entanglement axioms.
Assuming the global state is pure, the axiom {\bf A0} asserts $I(D_r : D_R^c)_\omega=0$ for sufficiently small radii $r<R$, and {\bf A1} asserts $I(D_r : D_R^c \mid A_{r,R})_\omega=0$ for all cones $A$ and sufficiently small radii $r<R$.
A gluing mechanism then extends these axioms to arbitrary radii $r<R$ \cite{Shi_2019,Shi_2020}.
However, these axioms are rather restrictive. In particular, they exclude chiral topologically ordered phases, which are generally expected to have a nonzero correlation length. Nevertheless, it is believed that suitable approximate versions of these axioms may provide a general framework for two-dimensional topologically ordered states; see, e.g., \cite{KimKitaevRanard2026}.
In an approximate version, {\bf A0} and {\bf A1} would only be required to hold in the limit of large regions.
Versions with specified decay rates have been discussed in \cite{yang_modular_2026}.
Theorem~\ref{thm:all-dim} suggests to consider a version without any specified decay rate. 
Take a concentric family $D_r$ of disks of radius $r>0$.
Let us say that the pure state $\omega$ satisfies approximate {\bf A0} if $\lim_{r\to \oo} \lim_{R\to \oo} I(D_r : D_R^c)_\omega =0$, and that it satisfies approximate {\bf A1} if, for all cones $A$, $\lim_{r\to\oo} \lim_{R\to\oo} I(D_r : D_R^c \mid A_{r,R})_\omega =0$.
By Theorem~\ref{thm:all-dim}, this approximate {\bf A1} axiom is equivalent to Haag duality for \emph{all} cones.
Our assumption of considering a pure state in the thermodynamic limit already entails approximate {\bf A0}. 
Indeed, it is just approximate {\bf A1} for the empty cone $A=\emptyset$, which satisfies Haag duality $M_{\emptyset}=\CC = B(\H)'$.
Since Haag duality is independent of the specific choice of $D_r$ (e.g., the shape or the center point), the approximate axioms do not depend on this choice either.
Now, when Haag duality holds, DHR sector theory constructs the braided monoidal C*-category of transportable supereselection sectors localized in cones \cite{DoplicherHaagRoberts1971,DoplicherHaagRoberts1974,BuchholzFredenhagen1982,naaijkens_localized_2011,Ogata_2022_tensorcat}.
In forthcoming work, we show that this framework also accommodates chiral topological phases, including the non-Abelian Ising phase of the Kitaev honeycomb model~\cite{KitaevHoneycombEO}.
Therefore, the approximate entanglement bootstrap axioms without specified decay already encode the fusion rules and braiding of topological charges.

% \lvl{
% Consequently, by the above, Haag duality for a fixed region $A$ is not invariant under finite-depth quantum circuits and, thus, not stable within a topological phase.
% }
% This confirms the \lvl{intuition} in
% \cite{Ogata_2022_tensorcat} that Haag duality is not preserved under \lvl{quasi-local unitaries} in general.
% To the best of our knowledge, our result gives the first proof of this fact.

% Consequently, our CMI $\lim_{R\to\infty} I(D_{r}:D_{R}^{c}\mid A_{r,R})$
% can likewise be altered by an FDQC, just as axiom \textbf{A1} of the entanglement
% bootstrap can. This is consistent with the observation in
% \cite{Ogata_2022_tensorcat} that Haag duality is not preserved under FDQCs in general.

The above observations that Haag duality for a single cone follows from an entanglement law of the form \eqref{eq:entanglement-law} or from the entanglement bootstrap axiom \textbf{A1} imply the Haag duality theorems proven in \cite{naaijkens_haag_2012,fiedler_haag_2015,ogata_haag_2025}.

\subsubsection{Failure of Haag duality for non-abelian anyons}

It has recently been shown that a state containing a single non-abelian anyon can violate single-cone Haag duality, even in its approximate form \cite{wallick_nonabelian_2026}. 
We now test this example using our criterion.
The GNS representation of a single non-abelian anyon can be obtained from that of the vacuum by creating an anyon pair and transporting one of the anyons to infinity \cite{naaijkens_localized_2011,naaijkens_haag_2012,fiedler_haag_2015}, see Fig.~\ref{fig:single-cone-anyon}. 
Ref.~\cite{Shi_2019} shows that if a finite, simply connected region $X$ contains an anyon $a$, the entanglement entropy of $X$ shifts by a fixed amount:
\begin{equation}
    S_{a}(X)-S(X)=\log d_{a}
\end{equation}
where $d_{a}$ is the quantum dimension of $a$.
In our setup, $D_{r}\cup A_{r,R}$ and $D_{r}\cup B_{r,R}$ each contain the anyon, while $A_{r,R}$ and $B_{r,R}$ do not. This gives
\begin{equation}
    I_{a}(D_{r}:D_{R}^c\mid A_{r,R})-I(D_{r}:D^c_{R}\mid A_{r,R})=2\log d_{a}.
\end{equation}
Since $d_{a}>1$ for any non-abelian anyon, if the vacuum state satisfies Haag duality, we have $\lim_{R\to\infty}I_{a}(D_{r}:D_{R}^c\mid A_{r,R})=2\log d_{a}>0$ for any fixed $r$ such that the anyon is contained in $D_r$.
Thus, a non-abelian anyon violates Haag duality for single cones, reproducing part of the result of \cite{wallick_nonabelian_2026}.

Another application of our theorem is that it makes obvious that Haag duality is unaltered if the region $A$ is altered in a finite way, which was previously observed in \cite{vanluijk2025largescalestructureentanglementquantum}.
Indeed, such a change is simply not seen by the conditional mutual information \eqref{eq:all-dim-inner-limit} if the radius $r$ is large enough.

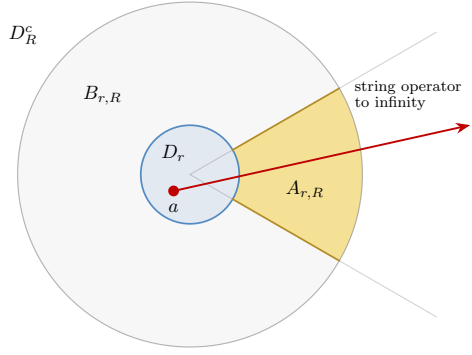
\begin{figure}[ht!]
  \centering
    \resizebox{0.4\linewidth}{!}{%
      \begin{tikzpicture}[
        line cap=round,
        line join=round,
        every node/.style={font=\small}
      ]

        % Geometry
        \def\rin{0.9}
        \def\rout{3.15}
        \def\theta{30}
        \def\raylen{5.2}

        \coordinate (O) at (0,0);

        % Anyon position
        \coordinate (anyon) at (-0.3,-0.3);

        % Same bounding box as panel (a)
        \path[use as bounding box]
          (-3.45,-3.35) rectangle (5.35,3.35);

        % Outer disk
        \fill[exteriorcolor, opacity=0.18]
          (O) circle (\rout);

        % Cone portion of the annulus: A_{r,R}
        \path[
          fill=conepale,
          draw=none
        ]
          (\theta:\rin)
          -- (\theta:\rout)
          arc[
            start angle=\theta,
            end angle=-\theta,
            radius=\rout
          ]
          -- (-\theta:\rin)
          arc[
            start angle=-\theta,
            end angle=\theta,
            radius=\rin
          ]
          -- cycle;

        % Inner disk D_r
        \fill[
          sphereblue,
          opacity=0.12
        ]
          (O) circle (\rin);

        % Light continuations of the cone rays
        \draw[
          wiregray,
          opacity=0.35,
          thin
        ]
          (O) -- (\theta:\raylen);

        \draw[
          wiregray,
          opacity=0.35,
          thin
        ]
          (O) -- (-\theta:\raylen);

        % Circular boundaries
        \draw[
          sphereblue,
          opacity=0.92,
          line width=0.9pt
        ]
          (O) circle (\rin);

        \draw[
          wiregray,
          opacity=0.55,
          semithick
        ]
          (O) circle (\rout);

        % Visible cone boundaries inside the shell
        \draw[
          coneedge,
          line width=0.9pt
        ]
          (\theta:\rin) -- (\theta:\rout);

        \draw[
          coneedge,
          line width=0.9pt
        ]
          (-\theta:\rin) -- (-\theta:\rout);

        % ========================================================
        % Anyon a
        % ========================================================
        \fill[red!75!black]
          (anyon) circle (2.8pt);

        \node at ($(anyon)+(0.0,-0.30)$)
          {$a$};

        % ========================================================
        % Semi-infinite string operator
        % ========================================================
        \draw[
          red!75!black,
          line width=1.0pt,
          ->,
          >=Stealth
        ]
          (anyon) -- (10:\raylen);

        % String label
        \node[
          align=left,
          font=\scriptsize,
          anchor=east
        ]
          at (5.10,1.45)
          {string operator\\[-1pt]to infinity};

        % Labels
        \node at (-0.30,0.40) {$D_r$};
        \node at (2.10,-0.28) {$A_{r,R}$};
        \node at (-1.60,1.45) {$B_{r,R}$};
        \node at (-3.05,2.55) {$D_R^c$};

      \end{tikzpicture}%
    }

  \caption{
    A state with an anyon $a$ contained in the inner disk $D_r$. The anyon is attached to a semi-infinite string operator supported inside the cone.
    }
  \label{fig:single-cone-anyon}
\end{figure}

\subsection{Translation-invariant spin chains and entropy density}

We consider a translation-invariant pure state $\omega$ on the quasi-local algebra $\A_\ZZ =\bigotimes_\ZZ M_d(\CC)$ of a quantum spin chain with local Hilbert space dimension $d$.
% We denote the subalgebra of the quasi-local algebra $\A_\ZZ$ associated with a region $A \subset \ZZ$ by $\A_{A}$.
For a finite interval $A\subset \ZZ$, translation invariance implies that the von Neumann entropy $S(A)_\omega$ of the reduced state only depends on the interval's length.
We set $S_n := S([0,n))_\omega$.
The entropy density of the state $\omega$ is 
\begin{equation}\label{eq:entropy-density}
    s := \lim_{n\to \oo} \frac{S_n }n.
\end{equation}
Existence of the limit \eqref{eq:entropy-density} is a well-known consequence of strong sub-additivity \cite{Fannes_2003}. 
Moreover, $s$ is finite: $0\le s \le \log(d)$.
For translation-invariant pure states, it is believed that the entropy density is zero, $s=0$.
This is known as the zero entropy density conjecture \cite{farkas_sharpness_2005,Fannes_2003,farkas_von_2007}.
A non-zero entropy density would amount to a volume law $S_n \sim s\cdot n$.
Strikingly, one can get arbitrarily close to a volume law as is shown in \cite{farkas_sharpness_2005}, which constructs translation-invariant pure states with arbitrarily fast sublinear growth.

We now use our abstract result Theorem \ref{thm:abstract} to show that the zero entropy density conjecture is equivalent to half-chain Haag duality.
We denote the half-chains by $ A =(-\oo,0]$ and $ B= (0,\oo)$.
We define finite approximations by $ A_n = (-n,0]$ and $ B_n = (0,n]$.
We define $D_n = A_n\cup B_n$, $A_{m,n} =A_n\setminus A_m$ and $B_{m,n} = B_n\setminus B_m$, so that
\begin{equation}
    D_n = (-n,n],\qquad
    A_{m,n} = (-n,-m], \qquad 
          B_{m,n} = (m,n]
    % \begin{aligned}
    %       D_n &=  A_n\cup  B_n = (-n,n] \qquad
    %      & D_n^c & =(-\oo,-n]\cup (n,\oo) \\
    %       A_{m,n}  &=  A_n \cap  C_m = (-n,-m] \quad &
    %       B_{m,n} &=  B_n \cap  C_m = (m,n].
    % \end{aligned}
\end{equation}
As for general spin systems, we consider the von Neumann algebras $M_X = \pi(\A_X)''$ in the GNS representation $(\H,\pi,\Omega)$, and we write $I(X : Y \mid Z):=I(M_X : M_Y \mid M_Z)_\Omega$.
We can now evaluate the conditional mutual information using eq.~\eqref{eq:CMI-entropy-rewrite}:
\begin{align}
    I(D_m : D_n^c \mid A_{m,n})_\omega
    & = S\big( (-n,m] \big)_\omega - S\big( (-n,-m] \big)_\omega + S\big( (-m,n] \big)_\omega - S\big( (m,n] \big)_\omega\nonumber \\
    & = 2 (S_{n+m}-S_{n-m}).
\end{align}
Next, we need to take the limit $n\to \oo$.
For this, we use the formula $\lim_{m\to \oo} (S_{m+1}-S_m) = s$, valid for all translation-invariant states \cite[Eq.~(62)]{fannes_fermionic_2012}, which implies

\begin{equation}
    \lim_{n\to \oo } \,(S_{n+m}- S_{n-m}) = 2ms
\end{equation}
via a telescoping sum.
Therefore, the conditional mutual information vanishes in the limit $n\to \oo$ for each $m$ if and only if the entropy density $s$ is zero.
Hence, our abstract result gives:

\begin{theorem}\label{thm:1d}
    Let $\omega$ be a translation-invariant pure state on $\A_\ZZ$.
    Then half-chain Haag duality holds in the GNS representation if and only if the entropy density is zero.
\end{theorem}

Directly checking half-chain Haag duality of a given translation-invariant pure state can be extremely hard.
In contrast, checking whether the entropy density is zero is often simple.
In fact, for large classes of models this is simply clear from the outset.
Moreover, whether the entropy density vanishes can, in principle, be studied numerically.\footnote{Note, however, that even when the entropy density is zero, the convergence rate $S_n/n \to 0$ can be arbitrarily bad \cite{farkas_sharpness_2005}. }

By Theorem~\ref{thm:1d}, the zero entropy density conjecture is equivalent to the half-chain Haag duality conjecture.
In footnote \ref{footnote:wrong-hd-proofs}, we mentioned two flawed proofs of the Haag duality conjecture.
Similarly, the preprint \cite{mohari_translation_2023} claims to resolve the zero entropy density conjecture, but its argument seems to have a serious gap.
% \footnote{
% The passage from Thm.~6.12 to Thm.~6.13, claims that the entropy density (denoted $s(\omega, \mathbb M_{loc})$ in \cite{mohari_translation_2023}) equals the entropic quantity $s(\omega)$ defined just before Thm.~6.12.
% However, the displayed argument based on the properties (a), (b), (c) only implies the inequality $s(\omega)\le s(\omega,\mathbb M_{loc})$.
% Thm.~6.12 claim that $s(\omega)$ equals the Connes-Størmer dynamical entropy, which is known to vanish for pure states, so the implication that is being proven simply amounts to the trivial one $0\le s$.
% }
Therefore, we regard both conjectures as open. 

\subsubsection{A strengthened Lieb-Schultz-Mattis constraint}

As an application of Theorem~\ref{thm:1d}, we derive an improved Lieb-Schultz-Mattis (LSM) constraint. 
Let $L = (-\infty,0]$ and $R = [1,\infty)$. 
A pure state $\omega$ on $\A_{\ZZ}$ satisfies the \emph{split property} if the half-chain algebras $M_{L}$ and, therefore, $M_R$ , are of type~$\mathrm{I}$ \cite{matsui2008spectralgapu1symmetrysplit,matsui2011boundedness}.
We denote the translation automorphism on $\A_\ZZ$ by $\tau$.

\begin{lemma}[{\cite[Cor.~4.2(i)]{matsui2008spectralgapu1symmetrysplit}}]\label{lem:matsui}
    Let $\omega$ be a translation-invariant pure state on $\A_{\ZZ}$ satisfying half-chain Haag duality. If $\omega$ satisfies the following clustering condition
    \begin{equation}\label{eq:C_j-Matsui}
        \sup_{\substack{a\in\A_L,\ b\in\A_R \\ \|a\|,\|b\|\le 1}}
        \bigl|\omega(a\,\tau^j(b)) - \omega(a)\,\omega(b)\bigr| \longrightarrow 0 \quad (j \to \infty),
    \end{equation}
    then $\omega$ satisfies the split property.
\end{lemma}

Matsui's proof assumes that every translation-invariant pure state satisfies half-chain Haag duality, which was claimed in \cite{KEYL_2008} but later it was found the proof is flawed \cite{matsui2011boundedness, van_luijk_entanglement_2025}; we therefore state it as an explicit assumption. The clustering condition in Eq.~\eqref{eq:C_j-Matsui} holds, for example, if there is a function $f\colon \mathbb{N} \to [0,\infty)$ with $\sum_{r\ge1} r f(r) < \infty$, e.g., $f(r)=O(r^{-2-\epsilon})$ for $\epsilon>0$, such that
\begin{equation}\label{eq:f-clustering}
    |\omega(ab) - \omega(a)\omega(b)| \le \|a\|\,\|b\| \sum_{x\in X,\, y\in Y} f(|x-y|)
\end{equation}
for all disjoint finite intervals $X, Y \subset \ZZ$ and all $a \in \A_X$, $b \in \A_Y$.%
\footnote{
Matsui denotes the left-hand side of Eq.~\eqref{eq:C_j-Matsui}, after taking the supremum, by $C_j$. There are exactly $r-j$ pairs $x \le 0$, $y \ge j+1$ with $y-x = r$, so $C_j \le \sum_{r > j} (r-j) f(r) \to 0$.}

Now let $G$ be a compact Lie group and $U$ a projective unitary representation of $G$ on $\mathbb{C}^d$ with a class in the (differentiable) group cohomology ${\rm{H}}^{2}(G,\mathrm{U(1))}$ \cite{brylinski2000differentiable}. 
The induced \emph{on-site action} $\beta:G\curvearrowright \A_{\ZZ}$ is given by $\beta(g) = \bigotimes_{x\in \ZZ} \mathrm{Ad}(U_x(g))$
%\mathrm{Ad}\big(\bigotimes_{x\in\Lambda} U_x(g)\big)(A)$ for $A \in \A_\Lambda$, for every finite subset $\Lambda\subset \ZZ$, 
where $U_x(g)$ is a copy of $U(g)$ acting on the site $x$. 
The symmetry action of $G\times \ZZ$, where $G$ acts on site via $\beta$ and $\ZZ$ acts by lattice translations, is called \textit{anomalous} if the projective representation $U$ defines a nontrivial class $[U]\in \mathrm{H}^{2}(G,\mathrm{U}(1)).$
This anomaly is a property of the symmetry action itself and is therefore independent of the choice of state.
The following lemma is proved in \cite{Ogata2019LSM,Ogata_2021,kapustin2024anomalous}.

\begin{lemma}\label{lem:LSM}
    Suppose the symmetry $G\times \ZZ$ defined above is anomalous. 
    If a pure state $\omega$ on $\A_{\ZZ}$ is invariant under $G\times \ZZ$, then $\omega$ does not satisfy the split property.
\end{lemma}
Now we derive the improved LSM theorem:

\begin{theorem}\label{thm:LSM}
    % Let $G\times\ZZ$ be the symmetry group described above, with nontrivial $[\sigma]$, and let $\omega$ be a pure state. 
    Let $\omega$ be a pure state.
    Then the following conditions cannot hold simultaneously:
    \begin{enumerate}
        \item The symmetry $G\times\ZZ$ is anomalous. %, i.e., $[U]$ is nontrivial in ${\rm H^{2}}(G,\mathrm{U}(1))$.
        \item $\omega$ is translation-invariant and symmetric under the on-site action $\beta$.
        \item $\omega$ has vanishing entropy density, i.e., $s=0$.
        \item $\omega$ satisfies the clustering condition in Eq.~\eqref{eq:C_j-Matsui}.
    \end{enumerate}
\end{theorem}

In particular, if the entropy density vanishes, then our LSM-type theorem says the clustering condition in Eq.~\eqref{eq:C_j-Matsui} must fail under these assumptions.
\begin{proof}[Proof of Theorem~\ref{thm:LSM}]
    Assume all properties are true.
    If $s = 0$, then $\omega$ satisfies half-chain Haag duality by Theorem~\ref{thm:1d}. 
    If, moreover, $\omega$ satisfies the clustering condition in Eq.~\eqref{eq:C_j-Matsui}, Lemma~\ref{lem:matsui} implies that $\omega$ satisfies the split property, contradicting Lemma~\ref{lem:LSM}.
\end{proof}

\section{A duality formula for the conditional entropy}
\label{sec:condEnt-duality}

We consider the mutual information and conditional entropy in cases where one system is a von Neumann algebra.
In this case, these entropic quantities are all defined in terms of the quantum relative entropy $D(\placeholder\|\placeholder)$, defined on pairs of normal positive linear functionals on von Neumann algebras \cite{Araki1976entropyI,ohya2004quantum}.
In this section, we derive a duality formula for the conditional entropy that we will need for our argument.

Let $M$ be a von Neumann algebra on a Hilbert space $\H$ and let $\F$ be a finite-dimensional Hilbert space.
We set $F= B(\F)$ and $d=\dim\F$.
Suppose $\Omega\in \F\ox\H$ is a unit vector.
We write $\omega_X = \ip{\Omega}{(\placeholder)\Omega}|_X$ and abbreviate $\omega_{F\ox1}$ as $\omega_F$ etc.
We define the conditional entropy as \cite{berta_quantum_2013}
\begin{align}
    % I(F : M)_\Omega &:= D(\omega_{F\ox M} \,\|\, \omega_F \ox \omega_M),\\[5.5pt]
    H(F \mid M)_\Omega 
    &:= - D(\omega_{F\ox M} \,\|\, \tr \ox\, \omega_M), % = I(F : M)_\Omega - \log(d),
\label{eq:def-condEnt}
\end{align}
where $\tr$ is the unnormalized trace on $F$.
We then have:

\begin{lemma}\label{lem:duality}
    Under the assumptions above,
    \begin{equation}\label{eq:abstract-condEnt-duality}
        H(F\mid M)_\Omega + H(F\mid M')_\Omega =0.
    \end{equation}
\end{lemma}

We will obtain Lemma~\ref{lem:duality} as an immediate consequence of \cite[Eq.~(14)]{HollandsVariational}, which follows immediately from the spatial derivative-based definition of the relative entropy \cite{ohya2004quantum} and Connes' duality theorem for spatial derivatives
\cite[Thm.~9]{connes_spatial_1980}.

\begin{proof}
    Set $A = F\ox M$ and $B = 1\ox M$ and consider the conditional expectation $\eps := \tau \ox \id : A\to B$, where $\tau= d^{-1}\tr$ is the tracial state on $F$.
    The entropy duality formula \cite[Eq.~(14)]{HollandsVariational} asserts
    \begin{equation}\label{eq:hollands-formula}
        D(\omega_A \,\|\,\omega_B\circ \eps ) + D(\omega_{B'} \,\|\,\omega_{A'}\circ \eps^{-1} ) = 0,
    \end{equation}
    where $\eps^{-1}:B'\to A'$ is the dual operator-valued weight \cite{Kosaki1986Index}.
    In our case, it is given by
    \begin{equation}
        \eps^{-1} = d \cdot(\tr \ox\, \id_{M'}) = d^2 \cdot (\tau\ox \id_{M'}).
    \end{equation}
    Hence, plugging in the formulas for $\eps$ and $\eps^{-1}$, eq.~\eqref{eq:hollands-formula} gives
    \begin{equation}
        D(\omega_{F\ox M} \,\|\, \tau \ox \omega_M) + D(\omega_{F\ox M'}\,\|\, \tau \ox \omega_{M'}) - 2\log(d)= 0.
    \end{equation}
    Since $\tr = d\tau$, we have $D(\omega_{F\ox X} \,\|\, \tau \ox \omega_X) -\log (d) = D(\omega_{F\ox X} \,\|\,\tr \ox\omega_X) = -H(F\mid X)_\Omega$, so the claim is shown.
\end{proof}

% Following \cite{ohya2004quantum}, the quantum relative entropy defining the conditional entropy is given by
% \begin{equation}\label{eq:spatial-derivative-def}
%     D(\omega_{F\ox M} \,\|\, \Tr \ox \omega_M)
%     =
%     \ip\Omega{ \log \frac{d (\tr\ox \omega_M) }{d \omega_{M'}}\Omega},
% \end{equation}
% where $\frac{d \phi_X}{d\psi_{X'}}$ denotes Connes' spatial derivative and where $\phi_X$ and $\psi_{X'}$ are normal positive linear functionals on $X,X'$, respectively (here, $X=F\ox M$ and $X'=1\ox M'$).
% As noted in \cite{ohya2004quantum}, \eqref{eq:spatial-derivative-def} is valid in any representation, without faithfulness assumptions on the states involved.

% Now suppose $L$ and $R$ are commuting factors on $\H$.

% \begin{lemma}
%     \begin{equation}
%         H(F\mid L)_\Omega + H(F\mid R)_\Omega \le  \log [L:R']_0.
%     \end{equation}
% \end{lemma}

\section{An abstract entropic characterization of Haag duality}\label{sec:proof-abstract}

Let $\H$ be a Hilbert space and let $A,B \subset B(\H)$ be commuting von Neumann algebras.
We assume that $A$ and $B$ are generated by increasing sequences of finite-dimensional matrix algebras $(A_n)_n$ and $(B_n)_n$ 
\begin{equation}
    A = \bigvee_n A_n,     
    \qquad\text{and}\qquad
    B = \bigvee_n B_n.
\end{equation}
Moreover, we assume that $A\vee B= B(\H)$.
Equivalently, $A$ and $B$ are tomographically complete \cite{van_luijk_uniqueness_2026}: A normal state on $B(\H)$ is uniquely determined by the $A|B$ correlations it induces, i.e., if $\omega_1(ab)=\omega_2(ab)$ for all $a=a^*\in A$, $b=b^*\in B$, then $\omega_1=\omega_2$.
As a consequence of tomographic completeness, $A$ and $B$ are factors, i.e., $A\cap A'=\CC=B\cap B'$.\footnote{
Indeed, $A\cap A' \subset B' \cap A' = (A\vee B)' = B(\H)'=\CC$ and similarly for $B$.}
We introduce the following notation:
\begin{equation}
    D_n := A_n\vee B_n \cong A_n\otimes B_n ,
    \qquad 
    C_n := D_n' = A_n'\cap B_n'. 
\end{equation}
Note that $D_n$ is a finite type I factor, whereas $C_n$ is type I$_\oo$ (assuming $\H$ is infinite dimensional).
Our notation is inspired by the geometry in Fig.~\ref{fig:2D-case}, where the algebras $D_n$ correspond to a sequence of disks with increasing radii.

Let $m\le n$. Since $A$ and $B$ commute, an operator $x\in A_n$ commutes with $A_m$ if and only if it commutes with $A_m\vee B_m=D_m$ or, equivalently, is contained in $C_m$.
The analogous statement holds when $A$ and $B$ are swapped.
Thus,
\begin{equation}\label{eq:relative-comm}
    A_{m,n} := A_m' \cap A_n = C_m \cap A_n,
    \qquad
    B_{m,n} := B_m' \cap B_n = C_m \cap B_n.
\end{equation}
For all $m \le n<\oo$, the four pairwise commuting type I factors $A_{m,n}, B_{m,n}, D_m, C_n$ jointly generate $B(\H)$:
\begin{equation}\label{eq:4-partite-system}
    A_{m,n} \vee B_{m,n} \vee C_n\vee D_m =B(\H).
\end{equation}
Indeed, $A_n = A_m \vee A_{m,n}$ and $B_n= B_m\vee B_{m,n}$, so that
% \begin{equation}
    $D_n = A_n\vee B_n 
    = A_m \vee A_{m,n} \vee B_m\vee B_{m,n} = D_m \vee A_{m,n} \vee B_{m,n}.$
% \end{equation}
Since $B(\H) = C_n \vee D_n$, this shows \eqref{eq:4-partite-system}.
Therefore, we can define the conditional mutual information in the usual way.
For a pure state $\Omega\in \H$, one has
\begin{equation}
    I(D_m : C_n \mid A_{m,n})_\Omega 
    = S(D_m \vee A_{m,n}) + S(C_n \vee A_{m,n}) - S(A_{m,n}) - S(D_m\vee  C_n\vee  A_{m,n}),
\end{equation}
where $S(X)$ denotes the von Neumann entropy of the reduced state $\omega_X := \ip\Omega{(\,\cdot\,)\Omega}|_X$.
We can now state our abstract Haag duality theorem:

\begin{theorem}\label{thm:abstract}
    The following are equivalent:
    \begin{enumerate}[(1)]
        \item\label{it:abstract1} Haag duality $A=B'$;
        \item\label{it:abstract2} 
        For every unit vector $\Omega\in \H$ and every $m$, $\lim_{n\to\oo} I(D_m : C_n \mid A_{m,n} )_\Omega =0$;
        \item\label{it:abstract3} For some unit vector $\Omega\in\H$, $\lim_{m\to\oo} \lim_{n\to\oo}\ 
        I(D_m : C_n \mid A_{m,n} )_\Omega =0$.
    \end{enumerate}
\end{theorem}

Like Haag duality, these entropic criteria are obviously invariant under exchanging $A\leftrightarrow B$.
Indeed, purity of the global state and \eqref{eq:4-partite-system} imply the duality relation
\begin{equation}
     I(D_m : C_n \mid A_{m,n})_\Omega = I(D_m : C_n \mid B_{m,n})_\Omega.
\end{equation}

For the proof of Theorem~\ref{thm:abstract}, we need some preparation.
We begin with the following Lemma.
We set $A_{m,\oo}:= A_m'\cap A = A \cap C_m$ and $B_{m,\oo} := B_m'\cap B = B\cap C_m$.
These equalities follow in the same way as in \eqref{eq:relative-comm}.

\begin{lemma}\label{lem:CMI-limit}
    Under the assumptions of Theorem~\ref{thm:abstract}, we have, for each $m$,
    \begin{equation}\label{eq:CMI-limit}
        \lim_{n\to\oo} I(D_m:C_n \mid A_{m,n})_\Omega
        = H( D_m \mid A_{m,\oo})_\Omega + H( D_m \mid B_{m,\oo})_\Omega.
    \end{equation}
\end{lemma}
\begin{proof}
    Since the global state is pure, since $D_m \vee A_{m,n}\vee B_{m,n} \vee C_n = B(\H)$, and since $A_{m,n} \vee B_{m,n} = (D_m\vee C_n)'$, we can write the conditional mutual information as a sum of conditional entropies
    \begin{align}\label{eq:sum-condEnt-equals-CMI}
        I(D_m : C_n \mid A_{m,n})_\Omega 
        &=  H(D_m \mid A_{m,n})_\Omega + H(D_m \mid B_{m,n})_\Omega.
    \end{align}
    For fixed $m$, we have
    \begin{equation}
        A_{m,\oo} = \bigvee_n A_{m,n}, \qquad B_{m,\oo} = \bigvee_n B_{m,n}.
    \end{equation}
    Recall that the conditional entropy is defined in terms of the quantum relative entropy $D(\placeholder\|\placeholder)$, see \eqref{eq:def-condEnt}.
    Therefore, the continuity of the quantum relative entropy under exhaustion \cite[Cor.~II.5.12]{ohya2004quantum} implies that the right-hand side of \eqref{eq:sum-condEnt-equals-CMI}
    approaches $H(D_m \mid A_{m,\oo})_\Omega + H(D_m \mid B_{m,\oo})_\Omega$ as $n\to\oo$.
\end{proof}

We introduce the notation
\begin{equation}
    \Hat A := B'.
\end{equation}
Note that $A\subset \Hat A$ and that Haag duality is equivalent to  $A=\Hat A$.
We denote the relative commutant of the inclusion $A_m\subset \Hat A$ by $\Hat A_{m,\oo} := A_m'\cap \Hat A$.
By the argument just before \eqref{eq:relative-comm}, and because $A_m$ is a type I factor, we have
\begin{equation}\label{eq:hat-well-behaved}
    \Hat A_{m,\oo} = C_m \cap \Hat A, 
    \qquad
    \Hat A = A_m \vee \Hat A_{m,\oo}.
\end{equation}
A helpful way to think of $\Hat A$ is given by the following characterization:
\begin{equation}
    \Hat A =B' = \bigg(\bigvee_n \,B_n\bigg)' =  \bigcap_n \,B_n',
\end{equation}
which says that $\Hat A$ contains every operator localized outside of every finite approximation of $B$.
Since $B_n' = (B_n'\cap D_n) \vee D_n' = A_n\vee C_n$, we can also write $\Hat A$ as $\Hat A =\cap_n (A_n\vee C_n)$.

We now consider the conditional entropy.
The duality formula for the conditional entropy, Lemma~\ref{lem:duality}, and eq.~\eqref{eq:hat-well-behaved} give
\begin{equation}\label{eq:condEnt-duality-applied}
    H(D_m \mid  B_{m,\oo})_\Omega =  - H(D_m \mid \Hat A_{m,\oo})_\Omega .
\end{equation}
Therefore, we can write the right-hand side of \eqref{eq:CMI-limit} as
\begin{equation}\label{eq:condEnt-difference}
    H(D_m \mid A_{m,\oo})_\Omega - H(D_m \mid \Hat A_{m,\oo})_\Omega.
\end{equation}
The entropic condition in Theorem~\ref{thm:abstract} is simply that \eqref{eq:condEnt-difference} vanishes as $m\to \oo$.
To understand this better, we relate \eqref{eq:condEnt-difference} to the mutual information. 
If $X\subset C_m$ is a von Neumann algebra, e.g., $A_{m,\oo}$ or $\Hat A_{m,\oo}$, the mutual information is defined as
\begin{equation}\label{eq:MI-def}
    I(D_m : X)_\Omega := D(\omega_{D_m \vee X} \,\|\,\omega_{D_m}\ox \omega_{X}).
\end{equation}
It relates to the conditional entropy via%
\footnote{
Indeed, one has
    $I(D_m:X)_\Omega
    = D(\omega_{D_m\vee X} \,\|\, \omega_{D_m}\ox\,\omega_X)
    = D(\omega_{D_m\vee X} \,\|\, \tr \ox\,\omega_X) - D(\omega_{D_m}\|\tr) = S(D_m)_\Omega - H(D_m\mid X)_\Omega$.
}
\begin{equation}
    I(D_m : X)_\Omega = S(D_m)_\Omega - H(D_m \mid X)_\Omega.
\end{equation}
Therefore, using the mutual information, one can write the difference of conditional entropies as
\begin{equation}\label{eq:condEnt-difference-DPI-defect}
    H(D_m \mid A_{m,\oo})_\Omega - H(D_m \mid \Hat A_{m,\oo})_\Omega
    =
    I(D_m : \Hat A_{m,\oo})_\Omega - I(D_m : A_{m,\oo})_\Omega.
\end{equation}
We know that this quantity is non-negative because it equals the limit of the conditional mutual information.
Another way to see the non-negativity is via the data-processing inequality \cite[Thm.~II.5.3]{ohya2004quantum} for the inclusion $D_m\vee A_{m,\oo}\subset D_m\vee \Hat A_{m,\oo}$.
Hence, if the conditional mutual information goes to zero as $n\to\oo$ followed by $m\to\oo$, then the defect in the data-processing inequality goes to zero as $m$ becomes large.
This allows us to apply an approximate recovery theorem, first shown in \cite{junge_universal_2018} and extended to von Neumann algebras in \cite{faulkner_approximate_2022,FaulknerHollands}.
We cannot immediately apply the theorem in \cite{faulkner_approximate_2022} because we cannot assume faithfulness of the marginal states of $\Omega$.
In Appendix~\ref{sec:appendix-recovery}, we derive the approximate recovery theorem as we need it, by reducing the statement to the faithful case and  then applying \cite{FaulknerHollands}.
Let 
\begin{equation}
    \alpha_m: \Hat A_{m,\oo} \to A_{m,\oo}
\end{equation}
be the modified Petz dual map of the inclusion $A_{m,\oo} \subset \Hat A_{m,\oo}$ with respect to the state $\omega|_{\Hat A_{m,\oo}}$, as defined in eq.~\eqref{eq:subunital-modified-recovery} of Appendix~\ref{sec:appendix-recovery}.
Then $\alpha_m$ is a normal subunital completely positive map with $\alpha_m(1) = \supp(\omega_{A_{m,\oo}})$.
Applying the approximate recovery theorem in  Lemma~\ref{lem:subunital-approx-recovery} to $\chi = \omega_{D_{m}\vee \Hat A_{m,\oo}}$, $\psi=\omega_{D_m}$ and $\phi = \omega_{\Hat A_{m,\oo}}$, gives
\begin{equation}\label{eq:approx-recovery}
    I(D_m:\Hat A_{m,\oo})_\Omega - I(D_m:A_{m,\oo})_\Omega \ge \frac14\,
    \|\omega_{D_m\vee \Hat A_{m,\oo}}- \omega_{D_m\vee A_{m,\oo}}\circ (\id_{D_m} \otimes \alpha_m)\|^2 .
\end{equation}
Note that Lemma~\ref{lem:subunital-approx-recovery} may be applied here since $\supp(\omega_{D_m\vee \Hat A_{m,\oo}}) \le \supp(\omega_{{D_m} \ox \omega_{\Hat A_{m,\oo}}}) =\supp(\omega_{D_m}) \cdot \supp(\omega_{\Hat A_{m,\oo}})$ and since $I(D_m : \Hat A_{m,\oo}) <\oo$.
The latter holds because $D_m$ is finite-dimensional.

\begin{proof}[Proof of Theorem \ref{thm:abstract}]
The implication \ref{it:abstract2} $\Rightarrow$ \ref{it:abstract3} is trivial.

\ref{it:abstract1} $\Rightarrow$ \ref{it:abstract2}:
Suppose Haag duality holds.
Then, by \eqref{eq:hat-well-behaved}, we also have $\Hat A_{m,\oo} = A_{m,\oo}$.
Therefore, \eqref{eq:condEnt-duality-applied} gives
\begin{equation}\label{eq:2D-sum-condEnt-infinite}
    H(D_m \mid A_{m,\oo} )_\Omega + H(D_m \mid  B_{m,\oo})_\Omega = 0,
\end{equation}
for all $m$.
By Lemma~\ref{lem:CMI-limit}, this shows the claim: $\lim_{n\to\oo} I(D_m:C_n \mid A_{m,n})_\Omega =0$ for each $m$.

\ref{it:abstract3} $\Rightarrow$ \ref{it:abstract1}:
By Lemma~\ref{lem:CMI-limit}, the assumption gives 
\begin{equation}
    \lim_{m\to\oo} \ H(D_m \mid A_{m,\oo})_\Omega + H(D_m \mid B_{m,\oo})_\Omega =0.
\end{equation}
As explained above, see eqs.~\eqref{eq:condEnt-duality-applied} to \eqref{eq:approx-recovery}, 
this implies
\begin{equation}\label{eq:approx-recovery-applied}
    \lim_{m\to\oo} \ \| \omega_{D_m \vee \Hat A_{m,\oo}} - \omega_{D_m \vee  A_{m,\oo}} \circ (\id_{D_m} \ox \alpha_m) \| =0,
\end{equation}
where $\alpha_m : \Hat A_{m,\oo} \to A_{m,\oo} $ is the modified Petz recovery map.
We will use \eqref{eq:approx-recovery-applied} to prove that a general element $x\in \Hat A$ is the weak limit of elements in $A$.
For $x \in \Hat A$, we set 
\begin{equation}
    x_m := (\id_{A_m}\ox \alpha_m)(x) \in A = A_m \vee A_{m,\oo},
\end{equation}
where we used \eqref{eq:hat-well-behaved}.
Note that $\|x_m\|\le \|x\|$ is uniformly bounded.
Thus, when checking weak convergence $x_m \to x$, it is enough to consider matrix elements from the dense subspace $\cup_{k} D_k\Omega$.
For $a,b \in D_k$ and $m\ge k$, we have
\begin{align*}
    | \langle a\Omega, (x - x_m)b\Omega\rangle |
    &= \big|\, \big\langle\Omega, \big( a^*xb - a^* (\id_{A_m} \ox\alpha_m)(x)b\big) \Omega\big\rangle \,\big|\\
    &= \big|\, \big\langle\Omega, \big( a^*xb - (\id_{D_m} \ox\alpha_m)(a^*xb)\big) \Omega\big\rangle\,\big| \\
    &= \big|\, \omega_{D_m\vee \Hat A_{m,\oo}}(a^*xb) - \omega_{D_m\vee  A_{m,\oo}} \circ (\id_{D_m}\ox \alpha_m) (a^*xb) \,\big|\\
    &\le
    \|\omega_{D_m \vee \Hat A_{m,\oo}}- \omega_{D_m \vee A_{m,\oo}}  \circ (\id\ox\alpha_m) \| \, \|a^*xb\|.
\end{align*}
By \eqref{eq:approx-recovery-applied}, the right-hand side goes to zero, showing that $x_m\to x$ weakly.
Since each $x_m$ is in $A$, the fact that $A$ is weakly closed shows $x\in A$.
Given that $x$ was arbitrary, this proves $A=\Hat A =B'$.

\end{proof}

The paper \cite{faulkner_asymptotically_2022} also proves an entropic criterion for Haag duality.
Their criterion uses smooth relative entropies and requires checking convergence on a dense subset of all states, whereas our criterion only involves a single state, see \cite[Thm.~2]{faulkner_asymptotically_2022}, in the context of holography. 
Their proof is different from ours, although it also builds on Petz's work on sufficiency in von Neumann algebras.

\paragraph{Acknowledgements.}
We thank Alexander Stottmeister, Daniel Wallick, Henrik Wilming, and Jinmin Yi for several fruitful discussions.
% We thank Stefan Hollands for a conversation about approximate recovery with non-faithful states. 

\paragraph{Author contributions and AI usage.}
AI played an important role in the discovery of the proof of the main theorem. 
The paper is entirely written by the authors. 
The two authors contributed equally and the author list is ordered alphabetically.
Below is a detailed account of the scientific contributions.

This work started as a project of RL and Jinmin Yi about translation-invariant pure states on spin chains, where AI found that Haag duality implies that the entropy density vanishes. 
Using RL's idea to use Petz recovery, it was then able to prove the converse, yielding Theorem~\ref{thm:1d}.
RL invited LvL to join the project, and the two of us worked through the proof in detail, which led to several improvements in the proof strategy, e.g., we introduced the conditional entropy formula \eqref{eq:abstract-condEnt-duality}, which allowed us to cut out a large section of the proof.
Analyzing the argument further, we found that a similar proof strategy applies to cones in 2D spin systems. 
LvL then abstracted the argument, yielding the purely algebraic version, which is now Theorem~\ref{introthm:abstract}. 
Apart from Theorem~\ref{thm:1d}, all applications to quantum spin systems were discovered and worked out by us.
Jinmin decided he did not want to be included as an author, feeling that he had not contributed much to the project.
Moreover, AI was used to proofread the manuscript. It found some mistakes in our arguments, which we then corrected.
AI was also used to create the figures.

\paragraph{Funding.}
Research at Perimeter Institute and the University of Waterloo is supported in part by the Government of Canada through the Department of Innovation, Science and Economic Development and by the Province of Ontario through the Ministry of Colleges and Universities.
RL is also supported by the Simons Collaboration on Global Categorical Symmetries through Simons Foundation grant 888996.

\null

\appendix
\numberwithin{equation}{section}

\section{Approximate Petz recovery without faithfulness}
\label{sec:appendix-recovery}

In this appendix we use the approximate recovery theorem in \cite{FaulknerHollands} to obtain an approximate recovery result without any faithfulness assumptions.
The earlier paper \cite{faulkner_approximate_2022} already discusses the non-faithful case, but their reduction to the faithful case seems to be incomplete. 

We consider $\sigma$-finite von Neumann algebras $N\subset M$ and a normal state $\phi$ on $M$.
We denote the support projections on $M$ and $N$ by $p = \supp(\phi)$ and $q = \supp(\phi|_N)$.
Then $p\le q$. 
We consider the corners $M_0:=pMp$ and $N_0:=qNq$.
Note that the restrictions $\phi|_{M_0}$ and $\phi|_{N_0}$ are faithful.
The inclusion $\iota:N\hookrightarrow M$ induces a normal unital completely positive map $\iota_0:N_0\to M_0$, $a\mapsto pap$.
Note that 
\begin{equation}
    (\phi|_{M_0} \circ \iota_0) (a) 
    = \phi(pap) = \phi(qaq)
    = \phi|_{N_0}(a), \qquad a\in N_0.
\end{equation}
Thus, we have a state-preserving normal unital completely positive map $\iota_0: (N_0, \phi|_{N_0}) \to (M_0,\phi|_{M_0})$.
Its modified Petz dual map \cite{junge_universal_2018,FaulknerHollands}, defined as an average of rotated Petz dual maps \cite[Eq.~(19)]{FaulknerHollands}, is a state-preserving normal unital completely positive map
\begin{equation}
    \alpha_0 : (M_0,\phi|_{M_0}) \to (N_0, \phi|_{N_0}).
\end{equation}
We use it to define a state-preserving normal completely positive subunital map
\begin{equation}\label{eq:subunital-modified-recovery}
    \alpha:= \alpha_0(p(\,\cdot\,)p) : (M,\phi) \to (N,\phi|_N), \qquad \alpha(1) = q.
\end{equation}
As a consequence of the approximate recovery theorem \cite{FaulknerHollands}, we have:

\begin{lemma}\label{lem:subunital-approx-recovery}
    Let $\psi$ be a normal state on a $\sigma$-finite von Neumann algebra $P$, and let $\chi$ be a normal state on $P\bar\otimes M$.
    If $\supp(\chi) \le \supp(\psi\ox\phi)$ and $D(\chi \| \psi \ox \phi)<\oo$,
    then
    \begin{equation}\label{eq:abstract-approx-rec}
        D(\chi \,\|\, \psi \ox \phi) - D(\chi|_{P\bar\otimes N} \,\|\, \psi \ox \phi|_N) \ge
        \frac14 \| \chi - \chi|_{P\bar\otimes N} \circ (\id_P\ox\alpha) \|^2.
    \end{equation}
\end{lemma}

\begin{proof}
    Let $e=\supp(\psi) \in P$ and set $P_0 := ePe$.
    Note that $\psi|_{P_0}$ is a faithful state.
    Consider the state-preserving normal unital completely positive map
    \begin{equation}
        \id_{P_0} \ox \iota_0  : (P_0\bar\ox N_0, \psi|_{P_0}\ox \phi|_{N_0}) \to (P_0\bar\ox M_0, \psi|_{P_0}\ox \phi|_{M_0}).
    \end{equation}
    It is clear from the definition of the Petz dual map \cite{ohya2004quantum} that the dual map of a tensor product map relative to faithful product states is just the tensor product of the Petz dual maps. 
    The same is true for the rotated Petz dual map, but it is false in general for the modified Petz dual map, which is defined as an average of rotated Petz dual maps.
    In our case, however, one of the maps is the identity map, which is its own rotated Petz dual.
    Therefore, the modified Petz dual map \cite[Eq.~(19)]{FaulknerHollands} of $\id_{P_0}\ox\iota_0$ is the state-preserving normal unital completely positive map
    \begin{equation}
         \id_{P_0} \ox \alpha_0 : (P_0\bar\ox M_0, \psi|_{P_0}\ox \phi|_{M_0}) \to (P_0\bar\ox N_0, \psi|_{P_0}\ox \phi|_{N_0}).
    \end{equation}
    Thus, the approximate recovery theorem \cite[Thm.~2, eq.~(22)]{FaulknerHollands} asserts
    \begin{multline}
        D(\chi|_{P_0\bar\ox M_0} \,\|\, \psi|_{P_0} \ox \phi|_{M_0}) - D(\chi|_{P_0\bar\otimes N_0} \,\|\, \psi|_{P_0} \ox \phi|_{N_0}) \\
         \ge
        \frac14 \| \chi|_{P_0\bar\ox M_0} - \chi|_{P_0\bar\otimes N_0} \circ (\id_{P_0}\ox\alpha_0) \|^2.\label{eq:cutdown-approx-rec}
    \end{multline}
    Since the quantum relative entropy $D(\placeholder\|\placeholder)$ is invariant under passing to the support corner of the second state, the left-hand side of eq.~\eqref{eq:cutdown-approx-rec} equals
    $
            D(\chi \,\|\, \psi \ox \phi) - D(\chi|_{P\bar\ox N} \,\|\, \psi \ox \phi|_{N}), 
    $
    while the right-hand side equals
    $
        \frac14 \| \chi - \chi|_{P\bar\otimes N} \circ (\id_{P}\ox\alpha) \|^2.
    $
    This concludes the proof.
\end{proof}

Eq.~\eqref{eq:abstract-approx-rec} can be sharpened by using the approximate recovery theorem in terms of the measured relative entropy, see \cite{hollands_trace-_2023}.

\printbibliography

\end{document}